\documentclass{article} 
\usepackage{fullpage}
\usepackage{amssymb,bm,mathrsfs,amsmath,amsthm,graphicx,subcaption,float, multirow}
\usepackage{caption}

\DeclareMathOperator*{\argmin}{\arg\!\min}

\newcommand{\true}[1]{{#1}^{\mathrm{(true)}}}

\newtheorem{lemma}{Lemma}
\newtheorem{theorem}{Theorem}

\title{Projected Wirtinger Gradient Descent for Spectral Compressed Sensing}

\author{Jian-Feng Cai
\thanks{Department of Mathematics, University of Iowa, Iowa City, IA 52242.
Email: \texttt{\{jianfeng-cai,suhui-liu\}@uiowa.edu}.}
\and Suhui Liu $^*$
\and Weiyu Xu
\thanks{Department of Electrical and Computer Engineering, University of Iowa, Iowa City, IA 52242.
Email: \texttt{weiyu-xu@uiowa.edu}.}
}

\begin{document}
\maketitle

\begin{abstract}
This paper considers reconstructing a spectrally sparse signal from a small number of randomly observed time-domain samples. The signal of interest is a linear combination of complex sinusoids at $R$ distinct frequencies. The frequencies can assume any continuous values in the normalized frequency domain $[0,1)$.  After converting the spectrally sparse signal recovery into a low rank structured matrix completion problem, we propose an efficient feasible point approach, named projected Wirtinger gradient descent (PWGD) algorithm, to efficiently solve this structured matrix completion problem. We further accelerate our proposed algorithm by a scheme inspired by FISTA. We give the convergence analysis of our proposed algorithms. Extensive numerical experiments are provided to illustrate the efficiency of our proposed algorithm. Different from earlier approaches, our algorithm can solve problems of very large dimensions very efficiently.
\end{abstract}

\section{Introduction}
Reconstructing a signal from a series of sampling measurements is a common theme in signal processing, which has numerous practical applications in radar, sonar, array processing, wireless communication, seismology, fluorescence microscopy, etc. Because of the constraints imposed by sampling hardware and physical measurement conditions, sometimes we can only obtain partial information, instead of full information, of a signal.  For example, when we try to infer the frequency components of a signal, we may only be able to get a small number of discrete time-domain samples of this signal. In this paper, the signal of interest is a weighted sum of 1-dimensional(1-D) complex sinusoids at $R$ distinct \emph{continuous} frequencies in the unit interval. From a small number of time-domain samples of the superposition of $R$ sinusoids, we are interested in recovering the complete signals, and identifying the existing frequencies. This signal model covers signals in various applications, for example, in acceleration of medical imaging \cite{LDP:MRM:07}, analog-to-digital conversion \cite{TLD:TIT:10}, and inverse scattering in seismic imaging \cite{BPTP:IP:02}.

Early conventional approaches, such as Prony's method \cite{Sch:BOOK:91}, ESPRIT \cite{RK:TASSP:89}, and the matrix pencil method \cite{HS:TASSP:90}, use sampling rates satisfying the Nyquist-Shannon sampling theorem. Compressed sensing (CS) is a new line of work in signal reconstruction, where, if the signal is sparse over some transform domain, the signal may be reconstructed with even fewer samples than the Nyquist sampling theorem requires \cite{CRT:TIT:06,Don:TIT:06}. In conventional compressed sensing, the signal of interest is generally assumed to have a sparse or approximately sparse representation over a finite discrete dictionary. However, signal parameters in practical applications often take values in a continuous domain. For example, in the problem considered in this paper, the frequencies take values in $[0,1)$. One can discretize the continuous signal parameters to a finite set of equi-spaced points, and then apply the theory of CS to recover the discretized parameters. However, when the discretization is not fine enough, this will cause \textit{basis mismatch} \cite{CLPCR:SP:11} in signal recovery. In basis mismatch, we will have non-negligible signal recovery errors resulting from the impact of discretization errors on CS signal recovery procedures, unless we make grid discretization very fine, leading to an undesirably large dictionary for signal recovery, to reduce signal recovery error \cite{TBR:Asilomar:2013}.

Recently there have been growing interests in designing new algorithms which can recover the continuous-valued parameters \emph{precisely} even from a small number of discrete nonuniform time samples. In \cite{CF:CPAM:14}, the authors proposed to use total variation minimization to find the continuous-valued frequencies from equi-spaced samples. In \cite{TBSR:TIT:13}, motivated by atomic norm minimization \cite{CRPW:FoCM:12}, the authors used atomic norm minimization to recover signal frequencies from nonuniform samples. In \cite{CF:CPAM:14} and \cite {TBSR:TIT:13}, the authors convert the signal frequency recovery into a low-rank Toeplitz matrix completion problem. In \cite{CC:TIT:14}, the problem of recovering signal frequencies from nonuniform samples is formulated as a low-rank Hankel matrix completion problem, inspired by Prony's method and the matrix pencil method. Though robust signal recovery is guaranteed theoretically through these methods in \cite{CF:CPAM:14,TBSR:TIT:13,CC:TIT:14}, convex optimization based low-rank structured matrix completions are not computationally efficient- the resulting optimization problems contain $O(N^2)$ unknowns explicitly, where $N$ is the dimension of signal.  To solve the resulting matrix completion problems, off-the-shelf algorithms such as SDPT3 \cite{TTT:OMS:99} use interior point methods which requires computing a Hessian matrix of size $O(N^4)$ in its Newton step. First-order methods, such as alternating direction method of multipliers (ADMM) and proximal point algorithm (PPA), need a dual matrix that is unstructured \cite{FPST:SIMAX:13}, and, consequently, these algorithms require memory of size $O(N^2)$. Therefore, these convex optimization approaches are not suitable for recovering signals of large dimensions.

To efficiently recover high-dimensional signals,  this paper proposes a projected Wirtinger gradient descent (PWGD) method for low-rank Hankel matrix completion. Instead of solving a convex relaxation of the low rank Hankel matrix completion problem, we directly deal with the non-convex low rank structured matrix completion problem. Our proposed PWGD algorithm is a feasible point algorithm, and it uses $O(NR)$ memory. Since the number of sinusoids, $R$,  is usually much smaller than $N$, the proposed algorithm provides efficient large scale signal recovery. Global convergence analysis of our algorithm is provided based upon Attouch and Bolte's theory \cite{ABRS:MOR:10,BST:MP:14}. To speed up our proposed algorithm, an acceleration technique scheme similar to FISTA \cite{BT:SIIMS:09} is given. The practical applicability of our algorithm is validated by numerical experiments, which show our algorithms can recover high-dimensional signals as a superposition of multiple sinusoids.

The paper is organized as follows. In Section \ref{sec:preliminaries}, we describe our signal model, give essential concepts about Hankel matrix, and formulate the signal recovery problem.  Our iterative algorithm and related convergence analysis is present in Section 3, where we also propose ways to accelerate the convergence of our algorithm. In Section 4, some numerical experiments are provided to demonstrate the performance of our algorithm. We then conclude our paper with a discussion of future work.

\section{Problem formulation}\label{sec:preliminaries}
In this section, we give some preliminaries on our signal model and the formulation of the signal reconstruction problem considered in this paper.

\subsection{Signal model}
The signal of our interest $\true{x}(t)$, $t\in\mathbb{R}$, is assumed as a linear combination of complex sinusoids at $R$ distinct frequencies $\true{f_k} \in [0,1)$ for $1 \leq k \leq R$, i.e.,
$$
\true{x}(t)=\sum_{k=1}^{R} \true{d}_k{e^{2\pi \imath \true{f}_k t}}, \quad t\geq0,
$$
where $\imath=\sqrt{-1}$.

Here the frequencies $\true{f}_k$'s are normalized to be in $[0,1)$ so that the signal can be uniquely determined by its time domain samples at integer points, and the associated coefficients $\true{d}_k$'s are the complex amplitudes. This model covers a wide range of signals in wireless communication, biology, automation, imaging science, seismology, etc.

To reconstruct the signal $\true{x}(t)$, early methods (e.g. Prony's method, the matrix pencil method, MUSIC) need time domain samples on uniformly sampled integer time points. More specifically, they use the following $(2N-1)$ samples in the time domain $\true{x}(t)$ at $t=0,1,\ldots, 2N-2$; and  then, in order to get the frequencies of $\true{x}(t)$, these early methods used linear algebra techniques involving linear structured matrices such as Hankel and Toeplitz matrices. However, due to physical measurement limitations, it is usually hard to get all the $2N-1$ samples of $\true{x}(t)$, $t=0,1,\ldots, 2N-1$, especially for signals with very high frequencies (before normalization) \cite{TLD:TIT:10}. So in this paper, we will consider non-uniform sampling in the time domain. We denote the underlying uniformly-sampled true signal as
$$
\true{\bm{x}}=[\true{x}(0),~\true{x}(1),\ldots,~\true{x}(2N-2)]^T\in\mathbb{C}^{2N-1},
$$
where $N$ is a large integer. However, we consider the case where only $M$ ($M<2N-1$) entries of $\true{\bm{x}}$ are observed. In this way, the sampling rate is significantly reduced. The same signal model is also considered in \cite{CC:TIT:14,TBSR:TIT:13,CRT:TIT:06}.

\subsection{Existing Algorithms}
Let $\Theta \subseteq \{0,1,\ldots,2N-2\}$ be the set of indices of observed entries of $\true{\bm{x}}$. Our goal is to reconstruct the true vector $\true{\bm{x}}$ from
\begin{equation}\label{eq:obs}
\bm{y}=\true{\bm{x}}_{\Theta}:=\{\true{x}(t)~|~t\in\Theta\}.
\end{equation}
There are several existing algorithms in the literature for recovering the $R$ sinusoids from the incomplete observations of $\true{\bm{x}}$.

One can discretize the frequency domain $[0,1)$ by uniform grid $\mathcal{G}$ with meshsize $1/(2N-1)$. Assume all frequencies $\true{f}_k$, $k=1,\ldots,R$, are on the grid $\mathcal{G}$. Then, the discrete signal $\true{\bm{x}}$ can be written as $\true{\bm{x}}=\bm{F}^*\bm{c}$, where $\bm{F}^*$ is the inverse of the discrete Fourier transform (DFT) matrix of order $2N-1$, and $\bm{c}\in\mathbb{C}^{2N-1}$ is a sparse vector with non-zero entries at indices $(2N-1)f_k$'s. Then, the samples \eqref{eq:obs} can be written as $\bm{y}=\bm{F}^*_{\Theta}\bm{c}$, where $\bm{F}^*_{\Theta}$ are partial rows of $\bm{F}^*$. Equivalently, our goal has turned into recovering the sparse vector $\bm{c}$. According to the theory of compressed sensing \cite{CRT:TIT:06}, when $\Theta$ is uniformly randomly drawn from all subsets of $\{0,1,\ldots,2N-2\}$ with cardinality $M$, the sparse vector $\bm{c}$ (hence $\true{\bm{x}}$) can be recovered exactly with high probability by solving
\begin{equation}\label{eq:l1fourier}
\min\limits_{\bm{c}}\|\bm{c}\|_1\quad\mbox{s.t.}~\bm{F}^*_{\Theta}\bm{c}=\bm{y},
\end{equation}
provided $M\geq O(K\log N)$. Efficient algorithms for solving \eqref{eq:l1fourier} include Bregman iterations \cite{COS:MC:09,COS:MC:09:2,YOGD:SIIMS:08} and iterative soft-thresholding algorithms \cite{DDD:CPAM:04,BT:SIIMS:09}. When the frequencies $\true{f}_k$'s are not on the grid $\mathcal{G}$, we expect to have a good approximation of $\true{\bm{x}}$ by solving \eqref{eq:l1fourier}, as the differences between the true frequencies and the grid $\mathcal{G}$ can be as small as $O(1/N)$. Unfortunately, this discretization method can lead to large recovery errors \cite{CLPCR:SP:11}. This phenomena is known as basis mismatch of compressed sensing. To overcome this limitation, we will consider frequencies $\true{f}_k$'s on the continuous domain $[0,1)$ instead of discretizing it with a meshsize of $1/(2N-1)$.

 Super-resolution compressed sensing \cite{CF:CPAM:14} and off-the-grid compressed sensing \cite{TBSR:TIT:13} consider reconstructing $\true{\bm{x}}$ from \eqref{eq:obs} under the assumption that the frequencies $\true{f}_k$'s take continuous values in the domain $[0,1)$. The authors of \cite{CF:CPAM:14} and \cite{TBSR:TIT:13} proposed to recover $\true{\bm{x}}$ by solving
\begin{equation}\label{eq:atommin}
\min_{\bm{u},\bm{x},t}~
\frac{u_0}{2(2N-1)}+\frac12t,
\quad\mbox{s.t.}~~
\left[\begin{matrix}\mathcal{T}(\bm{u})&\bm{x}\cr\bm{x}^*&t\end{matrix}\right]\succeq 0,
\end{equation}
where $\mathcal{T}$ is a linear operator that maps a vector $\bm{u}\in\mathbb{C}^{4N-3}$ to a $(2N-1)\times (2N-1)$ Toeplitz matrix $\mathcal{T}(\bm{u})$ satisfying $[\mathcal{T}(\bm{u})]_{jk}=u_{j-k}$ for all $0\leq j,k\leq 2N-2$. It was proved that, when $\Theta$ is uniformly randomly drawn from all the subsets of $\{0,1,\ldots,2N-2\}$ with cardinality $M\geq O(R\log(R/\delta)\log (N/\delta))$, the solution of \eqref{eq:atommin} will match the true discrete signal $\true{\bm{x}}$ with probability at least $1-\delta$. Extensions of the atomic norm minimization to 2D or higher dimensional complex exponentials can be found in \cite{XCMC:ITA:14,MCKX:ArXiv:14}.

Enhanced matrix completion \cite{CC:TIT:14} is another method that is able to reconstruct signals with frequencies taking continuous values. Enhanced matrix completion method converts the signal recovery problem to a Hankel matrix completion problem. Since this method is closely related to our proposed algorithm, we will introduce it in detail in the next subsection.

\subsection{Hankel Matrix Completion}
The enhanced matrix completion in \cite{CC:TIT:14} converts the reconstruction of $\true{\bm{x}}$ from \eqref{eq:obs} to a Hankel matrix completion problem. Let $\mathcal{H}$ be a linear operator that maps a vector in $\mathbb{C}^{2N-1}$ to a $N\times N$ Hankel matrix as follows
$$
\mathcal{H}:~\bm{x}\in\mathbb{C}^{2N-1} \longrightarrow \mathcal{H}\bm{x}\in\mathbb{C}^{N\times N},\qquad
[\mathcal{H}\bm{x}]_{jk}=x_{j+k},\quad 0\leq j,k\leq N-1.
$$
Define $\true{\bm{H}}=\mathcal{H}\true{\bm{x}}$. It can be checked that the rank of $\true{\bm{H}}$ is $R$, due to the following factorization
$$
\mathcal{H}\true{\bm{x}}=
\left[
\begin{matrix}
1&\ldots&1\cr
e^{2\pi\imath\true{f}_1}&\ldots&e^{2\pi\imath\true{f}_R}\cr
\vdots&\vdots&\vdots\cr
e^{2\pi\imath(M-1)\true{f}_1}&\ldots&e^{2\pi\imath(M-1)\true{f}_R}\cr
\end{matrix}
\right]
\left[\begin{matrix}
\true{d}_1\cr &\ddots\cr&&\true{d}_R
\end{matrix}
\right]
\left[
\begin{matrix}
1&e^{2\pi\imath\true{f}_1}\ldots&e^{2\pi\imath(M-1)\true{f}_1}\cr
\vdots&\vdots&\vdots\cr
1&e^{2\pi\imath\true{f}_R}\ldots&e^{2\pi\imath(M-1)\true{f}_R}\cr
\end{matrix}
\right]
$$
Then, instead of constructing the true signal $\true{\bm{x}}$ directly, we reconstruct the rank-$R$ Hankel matrix $\mathcal{H}\true{\bm{x}}$. Since $\mathcal{H}$ is one-to-one from a vector in $\mathbb{C}^{2N-1}$ to an $N\times N$ Hankel matrix, one can easily convert the reconstructed Hankel matrix back to a signal.

Now the signal reconstruction problem is formulated as
\begin{equation}\label{eq:min_knownrank}
\begin{aligned}
& \text{Find} & & \text{matrix } \bm{X}  \\& \text{subject to}& & \text{rank}(\bm{X}) \leq R, \\&&& X_{jk} = \true{H}_{jk}, \; (j,k) \in \Omega, \\&&& \bm{X} \text{ is a Hankel matrix},
\end{aligned}
\end{equation}
where $\Omega=\{(j,k)~|~j+k\in\Theta\}$ is the positions of known entries in $\true{\bm{H}}$. Since $\mathcal{H}$ is one-to-one from $\mathbb{C}^{2N-1}$ to the set of all $N\times N$ Hankel matrix, reconstructing $\true{\bm{x}}$ is equivalent to reconstructing $\true{\bm{H}}$. Following generic low-rank matrix completion \cite{CR:FoCM:09}, \eqref{eq:min_knownrank} is converted in \cite{CC:TIT:14} to a rank minimization problem and further relaxed to
\begin{equation}\label{eq:nucmin}
\min_{\bm{X}}~\|\bm{X}\|_*
\quad\mbox{s.t.}~~~
X_{jk}=\true{H}_{jk},~(j,k)\in\Omega,~~\mbox{and $\bm{X}$ is Hankel}.
\end{equation}
Here $\|\cdot\|_*$ is the sum of all the singular values, namely the nuclear norm. It was shown that, if $\Theta$ is uniformly randomly drawn from all subsets of $\{0,1,\ldots,2N-2\}$ with cardinality $M\geq O(R\log^4N)$, and certain separation conditions between frequencies are satisfied, then the solution of \eqref{eq:nucmin} recover $\true{\bm{H}}$ perfectly with dominant probability. Similar models are considered in \cite{CQXY:ArXiv:15}.

Though \eqref{eq:nucmin} is a convex optimization problem, there were no efficient ways to compute it for large problem dimensions. It has $O(N^2)$ explicit unknowns instead of $O(N)$ in $\true{\bm{x}}$. One may convert \eqref{eq:nucmin} to an SDP and then employ available packages such as SDPT3 \cite{TTT:OMS:99}. However, these packages use second-order methods, which require solving a huge linear system of order $O(N^2)\times O(N^2)$ at each step. Also, it is not straightforward \cite{FPST:SIMAX:13} to adapt nuclear norm minimization algorithms (e.g. \cite{CCS:SIOPT:10}) for generic low-rank matrix completion to solving \eqref{eq:nucmin}, as the Hankel constraint invokes $O(N^2)$ linear equality constraints. The semidefnite programming for atomic norm minimization \eqref{eq:atommin} suffers from the same issue of high computational complexity.

In this paper, instead of considering convex optimizations \eqref{eq:atommin} and \eqref{eq:min_knownrank}, we aim at attacking the original non-convex problem \eqref{eq:min_knownrank} directly. Non-convex algorithms has been proven to have the advantage of fast convergence in sparsity and low-rank reconstruction \cite{BD:ACHA:09,JMD:NIPS:10}. We propose an efficient algorithm based on projected Wirtinger gradient descent for this particular spectral signal recovery problem.

\section{Projected Wirtinger Gradient Algorithm}
In this section, we present our projected Wirtinger gradient algorithm, prove its convergence, and provide an acceleration scheme. Our basic algorithm is a projected gradient flow in the Wirtinger sense, and its convergence is obtained by applying the framework in \cite{ABRS:MOR:10} for proximal alternating minimization. To accelerate the convergence, we use the strategy used in FISTA \cite{BT:SIIMS:09}.

\subsection{Basic algorithm}
This section is devoted to presenting our basic algorithm for solving \eqref{eq:min_knownrank}. Let us define the set of all complex-valued matrices with rank no greater than $R$ as
\begin{equation}\label{eq:setR}
\mathscr{R}^R_{\mathbb{C}}= \{\bm{L}\in \mathbb{C}^{N \times N}|\mathrm{rank}(\bm{L})\leq R\}.
\end{equation}
Similarly, define the set of all complex-valued Hankel matrices that are consistent with the observed data
\begin{equation}\label{eq:setH}
\mathscr{H}=\{\mathcal{H}\bm{x}~|~\bm{x}\in\mathbb{C}^{2N-1},~\bm{x}_{\Theta}=\true{\bm{x}}_{\Theta}\}.
\end{equation}
The set $\mathscr{R}^R_{\mathbb{C}}$ is a smooth manifold and $\mathscr{H}$ is an affine space. Then, our signal recovery problem and also \eqref{eq:min_knownrank} can be formulated as the following optimization problem
\begin{equation}\label{eq:minRH}
\min_{\bm{L} \in \mathscr{R}^R_{\mathbb{C}},\bm{H}\in \mathscr{H}}~~ \frac{1}{2}{\|\bm{L}-\bm{H}\|^2_F}
\end{equation}

We will employ a projected gradient descent algorithm to solve \eqref{eq:minRH}. The objective $F(\bm{L},\bm{H}):=\frac{1}{2}\|\bm{L}-\bm{H}\|^2_F$ is a real-valued function with complex variables, which is not differentiable in the ordinary complex calculus sense. Nevertheless, $F(\bm{L},\bm{H})$ is differentiable with respect to the real and imaginary parts of its variables. Thus, our gradient flow is performed on the real and imaginary parts respectively. Denote
$$
\bm{Z}=\left[\begin{matrix}\bm{L}\cr\bm{H}\end{matrix}\right]=\Re+\imath\Im
$$
where $\Re$ and $\Im$ are the real and imaginary parts of $\bm{Z}$. Rewrite $F$ as $F(\Re,\Im)$. Then, in our gradient flow algorithm, $\Re$ is updated by $\frac{\partial F}{\partial\Re}$ and $\Im$ by $\frac{\partial F}{\partial\Im}$. In other words, $\bm{Z}$ is updated by $\frac{\partial F}{\partial\Re}+\imath\frac{\partial F}{\partial\Im}$. By Wirtinger calculus \cite{Fis:BOOK:05}, we have the relation
$$
\frac{\partial F}{\partial\Re}+\imath\frac{\partial F}{\partial\Im} =2\frac{\partial{F}}{\partial{\overline{\bm{Z}}}}.
$$
Direct calculations give
$$
2\frac{\partial{F}}{\partial{\overline{\bm{Z}}}}=
\left[\begin{matrix}2\frac{\partial{F}}{\partial\overline{\bm{L}}}\cr
2\frac{\partial{F}}{\partial\overline{\bm{H}}}\end{matrix}\right]
=\left[\begin{matrix}\bm{L}-\bm{H}\cr
\bm{H}-\bm{L}\end{matrix}\right].
$$
Using the Wirtinger gradient, our proposed algorithm is given as follows: at iteration $t$, we have
\begin{equation}\label{eq:alg_PWGD}
 \left\{ \begin{array}{l l}
 \bm{L}_{t+1}\in \mathcal{P}_{\mathscr{R}^R_\mathbb{C}}(\bm{L}_t-\delta_1(\bm{L}_t-\bm{H}_t)),\\
     \bm{H}_{t+1} \in \mathcal{P}_{\mathscr{H}}(\bm{H}_t -\delta_2(\bm{H}_t-\bm{L}_{t+1})),\\
    \end{array} \right.
\end{equation}
where $\delta_1>0$ and $\delta_2>0$ are step sizes, and $\mathcal{P}_{\mathscr{R}^R_\mathbb{C}}$ and $\mathcal{P}_{\mathscr{H}}$ are projections onto $\mathscr{R}^R_\mathbb{C}$ and $\mathscr{H}$ respectively. We call \eqref{eq:alg_PWGD} projected Wirtinger gradient descent (PWGD).

It remains to find out $\mathcal{P}_{\mathscr{R}^R_\mathbb{C}}$ and $\mathcal{P}_{\mathscr{H}}$ respectively. Since $\mathcal{P}_{\mathscr{R}^R_\mathbb{C}}(\bm{X})$ is the best rank-$R$ approximation to $\bm{X}$, according to Eckhart-Young Theorem \cite{GV:BOOK:96},
$$
\mathcal{P}_{\mathscr{R}^R_\mathbb{C}}(\bm{X})=\bm{U}_R\bm{\Sigma}_R\bm{V}_R^*,
$$
where the columns of $\bm{U}_R$ and $\bm{V}_R$ are the first $R$ left and right singular vectors of $\bm{X}$ respectively and $\bm{\Sigma}_R$ is a diagonal matrix with diagonals corresponding singular values.
The closed form of $\mathcal{P}_{\mathscr{H}}$ is given by the following lemma
\begin{lemma}\label{lem:projH}
We have
\begin{equation}\label{eq:projH}
\mathcal{P}_{\mathscr{H}}(\bm{X})=\mathcal{H}\bm{z},
\quad\mbox{where}\quad
z_j=
\begin{cases}
\true{x}_j,&\mbox{if }j\in\Theta,\cr
\mathrm{mean}\{X_{kl}~|~k+l=j\},&\mbox{otherwise.}
\end{cases}
\end{equation}
\end{lemma}
\begin{proof}
$\mathcal{P}_{\mathscr{H}}(\bm{X})$ is the solution of the following least square problem
\begin{equation*}\label{eq:lem:projH:10}
\begin{split}
\mathcal{P}_{\mathscr{H}}(\bm{X})&=\argmin_{\bm{Z}}\{\|\bm{Z}-\bm{X}\|_F^2~:~\bm{X}\in\mathscr{H}\}
=\mathcal{H}\cdot\argmin_{\bm{z}}\{\|\mathcal{H}\bm{z}-\bm{X}\|_F^2~:~\bm{z}_{\Theta}=\true{\bm{x}}_{\Theta}\}\cr
&=\mathcal{H}\cdot\argmin_{\bm{z}}\{\sum_{j=0}^{2N-2}\sum_{k+l=j}(z_j-X_{kl})^2~:~\bm{z}_{\Theta}=\true{\bm{x}}_{\Theta}\}.
\end{split}
\end{equation*}
It is obvious that the solution of the optimization problem in the last line is given by $z_j$ in \eqref{eq:projH}.
\end{proof}

The proposed PWGD algorithm \eqref{eq:alg_PWGD} is a feasible point algorithm. The iterates $\bm{L}_t$ and $\bm{H}_t$ are always in their feasible sets $\mathscr{R}^R_\mathbb{C}$ and $\mathscr{H}$ respectively. This property can significantly reduce the computational cost and storage, when $R$ is small compared to $N$. Since $\bm{L}_t\in\mathscr{R}^R_\mathbb{C}$, it is stored in a factorization form and only $O(NR)$ memory is necessary. Also, the Hankel matrix $\bm{H}_t$ can be represented by its parameters, which is of size only $O(N)$. Furthermore, in Step 1 of \eqref{eq:alg_PWGD}, it needs to compute only the first $R$ singular values and their corresponding singular vectors of $\bm{L}_t-\delta_1(\bm{L}_t-\bm{H}_t)$ in the computation of the projection. This can be done by, e.g., Krylov subspace methods, which invokes only the matrix-vector product of $\bm{L}_t-\delta_1(\bm{L}_t-\bm{H}_t)$. For the matrix-vector product of $\bm{L}_t$, since $\bm{L}_t$ is rank $R$ and in a factorization form, it can be done in $O(NR)$ operations. The matrix-vector product of the Hankel matrix $\bm{H}_t$ is implemented by fast Fourier transform \cite{GV:BOOK:96}, which needs only $O(N\log N)$ operations. Step 2 of \eqref{eq:alg_PWGD} needs averages of $\bm{L}_{t+1}$ along anti-diagonals.

\subsection{Convergence}
In this subsection, we prove the convergence of the proposed PWGD algorithm \eqref{eq:alg_PWGD}. Our proof is achieved by applying the convergence result in \cite{ABRS:MOR:10}.

Consider a general non-convex optimization problem
\begin{equation}\label{eq:gmin}
\min_{\bm{x},\bm{y}} \psi(\bm{x},\bm{y}):=\phi(\bm{x},\bm{y})+\theta(\bm{x})+\omega(\bm{y}),
\end{equation}
where the functions $\theta~:~\mathbb{R}^n\mapsto\mathbb{R}\cup\{+\infty\}$ and $\omega~:~\mathbb{R}^m\mapsto\mathbb{R}\cup\{+\infty\}$ are proper lower semicontinuous functions and $\phi~:~\mathbb{R}^n\times\mathbb{R}^m\mapsto\mathbb{R}$ is a $C^1$ function. It was proposed in \cite{ABRS:MOR:10} a proximal alternating minimization algorithm for solving \eqref{eq:gmin}
\begin{equation}\label{eq:PAMA}
\begin{cases}
\bm{x}_{k+1}\in\argmin_{\bm{x}\in\mathbb{R}^n}\psi(\bm{x},\bm{y}_k)+\frac{1}{2\lambda_k}\|\bm{x}-\bm{x}_k\|_2^2,\cr
\bm{y}_{k+1}\in\argmin_{\bm{y}\in\mathbb{R}^m}\psi(\bm{x}_{k+1},\bm{y})+\frac{1}{2\mu_k}\|\bm{y}-\bm{y}_k\|_2^2.
\end{cases}
\end{equation}
Under the assumption that the function $\psi$ satisfies the so-called Kurdyka-Lojasiewicz (KL) condition and $\nabla\phi$ is Lipschitz on bounded sets, \cite{ABRS:MOR:10} proved the convergence of \eqref{eq:PAMA}. Generally, the KL condition is not easy to check. A sufficient condition to guarantee the KL condition is the semi-algebraic property. A proper and lower semi-continuous function is called semi-algebraic if its graph is a semi-algebraic set. Recall a subset $S\subset\mathbb{R}^d$ is a real semi-algebraic set if there exists a finite number of real polynomial function $g_{ij},h_{ij}:~\mathbb{R}^d\mapsto\mathbb{R}$ such that
$$
S=\bigcup_{j=1}^p\bigcap_{i=1}^q\left\{\bm{u}\in\mathbb{R}^d~\big|~g_{ij}(\bm{u})=0,~h_{ij}(\bm{u})<0\right\}.
$$

Choose $\theta=\delta_{C}$ and $\omega=\delta_D$ are indicator functions for the sets $C\in\mathbb{R}^n$ and $D\in\mathbb{R}^m$ respectively. Recall the indicator function $\delta_C$ of a set $C$ is defined as $\delta_C(\bm{x})=\begin{cases}0,&\mbox{if }\bm{x}\in C,\cr +\infty,&\mbox{if }\bm{x}\not\in C\end{cases}$. Let $\phi(\bm{x},\bm{y})=\frac12\|\bm{x}-\bm{y}\|_2^2$. Then \eqref{eq:PAMA} becomes an alternating projection algorithm
\begin{equation}\label{eq:altproj}
\begin{cases}
\bm{x}_{k+1}\in\mathcal{P}_C\left(\bm{x}_k-\frac{1}{1+\delta_k}(\bm{x}_k-\bm{y}_k)\right),\cr
\bm{y}_{k+1}\in\mathcal{P}_D\left(\bm{y}_k-\frac{1}{1+\mu_k}(\bm{y}_k-\bm{x}_{k+1})\right).
\end{cases}
\end{equation}
The results in \cite{ABRS:MOR:10} imply the following convergence theorem of \eqref{eq:altproj}, which is a corollary of Corollary 12 of \cite{ABRS:MOR:10} and Theorem 3 and Example 2 of \cite{BST:MP:14}.
\begin{theorem}\label{thm:convPAMA}
Assume that the sets $C\subset\mathbb{R}^n$ and $D\subset\mathbb{R}^m$ are semi-algebraic. Let $(\bm{x}_k,\bm{y}_k)$ be generated by \eqref{eq:altproj} with $0<a<\delta_k,\mu_k<b$ for all $k$.
\begin{itemize}
\item[(a)] Either $\|(\bm{x}_k,\bm{y}_k)\|_2\to\infty$ as $k\to\infty$, or $(\bm{x}_k,\bm{y}_k)$ converges to a critical point of $\psi$.
\item[(b)] If we further assume $(\bm{x}_0,\bm{y}_0)$ is feasible and sufficiently close to a global minimizer of $\psi$, then $(\bm{x}_0,\bm{y}_0)$ converges to a global minimizer of $\psi$.
\end{itemize}
\end{theorem}

Next we apply Theorem \ref{thm:convPAMA} to the PWGD algorithm \eqref{eq:alg_PWGD} to get its convergence. The PWGD algorithm \eqref{eq:alg_PWGD} is in the same form as \eqref{eq:altproj}. However, our PWGD algorithm is performed in complex-valued matrix spaces, while the setting of Theorem \ref{thm:convPAMA} is in real. Nevertheless, we can identify any complex-valued matrix to a real one by concatenating its real and imaginary parts. Actually, as aforementioned, our Writinger gradient descent is exactly obtained in this way by considering the gradient with respect to the real and imaginary parts. Since the objective function $F(\bm{L},\bm{H})$ in \eqref{eq:minRH} does not change after this identification, we only to check the sets $\mathscr{R}^R_{\mathbb{C}}$ in \eqref{eq:setR} and $\mathscr{H}$ in \eqref{eq:setH} are semi-algebraic when viewed as sets of real and imaginary parts. This is done by the following two lemmas.
\begin{lemma}\label{lem:setR}
The set $\mathscr{S}_R$ defined as follows is a semi-algebraic set
$$
\mathscr{S}_R=\left\{[\bm{X},\bm{Y}]~\big|~(\bm{X}+\imath\bm{Y})\in\mathscr{R}^R_{\bm{C}}\right\}.
$$
\end{lemma}
\begin{proof}
Denote
$$
\mathscr{P}_r=\{[\bm{X},\bm{Y}]~|~\bm{X},\bm{Y}\in\mathbb{R}^{N\times N},~\mathrm{rank}(\bm{X}+\imath\bm{Y})=r\}
$$
and
$$
\mathscr{Q}_r=\left\{[\bm{X},\bm{Y}]~|~\bm{X},\bm{Y}\in\mathbb{R}^{N\times N},~\mathrm{rank}\left(\left[\begin{matrix}\bm{X}&-\bm{Y}\cr\bm{Y}&\bm{X}\end{matrix}\right]\right)=2r\right\}.
$$
We first prove $\mathscr{P}_r=\mathscr{Q}_r$ by showing $\mathscr{P}_r\subset\mathscr{Q}_r$ and $\mathscr{Q}_r\subset\mathscr{P}_r$ respectively. Let $[\bm{X},\bm{Y}]\in\mathscr{P}_r$, and a singular value decomposition (SVD) of $\bm{X}+\imath\bm{Y}$ is $(\bm{X}+\imath\bm{Y})=(\bm{U}_{\mathrm{Re}}+\imath \bm{U}_{\mathrm{Im}})\bm{\Sigma}(\bm{V}_{\mathrm{Re}}+\imath \bm{V}_{\mathrm{Im}})^*$, where $\bm{U}_{\mathrm{Re}},\bm{U}_{\mathrm{Im}},\bm{V}_{\mathrm{Re}},\bm{V}_{\mathrm{Im}}\in\mathbb{R}^{N\times r}$ and $\bm{\Sigma}\in\mathbb{R}^{r\times r}$. Then, by direct calculation, we see that an SVD of $\left[\begin{matrix}\bm{X}&-\bm{Y}\cr\bm{Y}&\bm{X}\end{matrix}\right]$ is given by
\begin{equation}\label{eq:SVDreal}
\left[\begin{matrix}\bm{X}&-\bm{Y}\cr\bm{Y}&\bm{X}\end{matrix}\right]
=\left[\begin{matrix}
  \bm{U}_{\mathrm{Re}} & -\bm{U}_{\mathrm{Im}} \cr
  \bm{U}_{\mathrm{Im}}  & \bm{U}_{\mathrm{Re}} \cr
  \end{matrix}\right]
  \left[\begin{matrix}
  \bm{\Sigma} & \bm{0} \cr
  \bm{0}  & \bm{\Sigma} \cr
  \end{matrix}\right]
\left[\begin{matrix}
  \bm{V}_{\mathrm{Re}} & -\bm{V}_{\mathrm{Im}} \cr
  \bm{V}_{\mathrm{Im}}  & \bm{V}_{\mathrm{Re}} \cr
  \end{matrix}\right]^*.
\end{equation}
Therefore, $\mathrm{rank}\left(\left[\begin{matrix}\bm{X}&-\bm{Y}\cr\bm{Y}&\bm{X}\end{matrix}\right]\right)=2r$, which implies $[\bm{X},\bm{Y}]\in\mathscr{Q}_r$ and further $\mathscr{P}_r\subset\mathscr{Q}_r$. Conversely, let $[\bm{X},\bm{Y}]\in\mathscr{Q}_r$. If $\left(\sigma,\left[\begin{matrix}\bm{u}_1\cr\bm{u}_2\end{matrix}\right],\left[\begin{matrix}\bm{v}_1\cr\bm{v}_2\end{matrix}\right]\right)$
is a singular triplet of $\left[\begin{matrix}\bm{X}&-\bm{Y}\cr\bm{Y}&\bm{X}\end{matrix}\right]$, then
$\left(\sigma,\left[\begin{matrix}-\bm{u}_2\cr\bm{u}_1\end{matrix}\right],\left[\begin{matrix}-\bm{v}_2\cr\bm{v}_1\end{matrix}\right]\right)$
is too by direct calculation. Therefore, the multiplicity of each singular value is even, and SVD's of $\left[\begin{matrix}\bm{X}&-\bm{Y}\cr\bm{Y}&\bm{X}\end{matrix}\right]$ must be in the form of \eqref{eq:SVDreal}. Consequently,  $(\bm{X}+\imath\bm{Y})=(\bm{U}_{\mathrm{Re}}+\imath \bm{U}_{\mathrm{Im}})\bm{\Sigma}(\bm{V}_{\mathrm{Re}}+\imath \bm{V}_{\mathrm{Im}})^*$ is an SVD of $(\bm{X}+\imath\bm{Y})$, which implies $\mathrm{rank}(\bm{X}+\imath\bm{Y})=r$. Therefore, $[\bm{X},\bm{Y}]\in\mathscr{P}_r$. Thus, $\mathscr{Q}_r\subset\mathscr{P}_r$.

Since $\mathscr{Q}_r$ is the intersection of the set of all rank-$2r$ real-valued matrices and the linear subspace of matrices in the form of $\left[\begin{matrix}\bm{X}&-\bm{Y}\cr\bm{Y}&\bm{X}\end{matrix}\right]$, it is deducted from \cite[Example 2]{BST:MP:14} that $\mathscr{Q}_r$ is a semi-algebraic set. This together with $\mathscr{P}_r=\mathscr{Q}_r$ implies $\mathscr{P}_r$ is a semi-algebraic set too.

Finally, it is obvious that $\mathscr{S}_R=\bigcup_{r=0}^{R}\mathscr{P}_r$. Therefore, $\mathscr{S}_R$ is a semi-algebraic set.
\end{proof}

\begin{lemma}\label{lem:setH}
The set $\mathscr{K}$ defined as follows is a semi-algebraic set
$$
\mathscr{K}=\left\{[\bm{X},\bm{Y}]~\big|~(\bm{X}+\imath\bm{Y})\in\mathscr{H}\right\}.
$$
\end{lemma}
\begin{proof}
Since $\mathcal{H}$ is a linear operator,
$$
\mathcal{H}\bm{x}=\mathcal{H}\Re(\bm{x})+\imath\mathcal{H}\Im(\bm{x}).
$$
Futher, for any $\bm{x}$ satisfying $\bm{x}_{\Theta}=\true{\bm{x}}_{\Theta}$, we have $\Re(\bm{x}_{\Theta})=\Re(\true{\bm{x}}_{\Theta})$ and $\Im(\bm{x}_{\Theta})=\Im(\true{\bm{x}}_{\Theta})$. Therefore,
$$
\mathscr{H}=\Re(\mathscr{H})+\imath\Im(\mathscr{H})
=\mathscr{K}_1+\imath\mathscr{K}_2,
$$
where
$$
\mathscr{K}_1=\{\mathcal{H}\bm{r}~|~\bm{r}\in\mathbb{R}^{2N-2},~\bm{r}_{\Theta}=\Re(\true{\bm{x}}_{\Theta})\},
\quad
\mathscr{K}_2=\{\mathcal{H}\bm{i}~|~\bm{i}\in\mathbb{R}^{2N-2},~\bm{i}_{\Theta}=\Im(\true{\bm{x}}_{\Theta})\}.
$$
This shows $\mathscr{K}=\mathscr{K}_1\times \mathscr{K}_2$. Since both $\mathscr{K}_1$ are $\mathscr{K}_2$ are affine spaces, their product $\mathscr{K}$ is also, which implies $\mathscr{K}$ is semi-algebraic.
\end{proof}

Combining Theorem \ref{thm:convPAMA} and Lemmas \ref{lem:setR} and \ref{lem:setH} leads to the following convergence results of the proposed algorithm \eqref{eq:alg_PWGD}.
\begin{theorem}\label{thm:convPWGD}
Let $(\bm{L}_t,\bm{H}_t)$ be generated by \eqref{eq:alg_PWGD} with $0<\delta_1,\delta_2<1$.
\begin{itemize}
\item[(a)] Either $\|(\bm{L}_t,\bm{H}_t)\|_F\to\infty$ as $t\to\infty$, or $(\bm{L}_t,\bm{H}_t)$ converges.
\item[(b)] If we further assume $(\bm{L}_0,\bm{H}_0)$ is feasible and sufficiently close to a global minimizer of $\{\frac12\|\bm{L}-\bm{H}\|_F^2|\bm{L}\in\mathscr{R}^R_{\mathbb{C}},\bm{H}\in\mathscr{H}\}$, then $(\bm{L}_0,\bm{H}_0)$ converges to a global minimizer of $\min_{\bm{L}\in\mathscr{R}^R_{\mathbb{C}},\bm{H}\in\mathscr{H}}\frac12\|\bm{L}-\bm{H}\|_F^2$.
\end{itemize}
\end{theorem}

We would like to remark that the unboundedness in \eqref{thm:convPWGD}(a) is not a problem and can be overcome by introduce a bound constraint in the set $\mathscr{H}$. For example, we can define $\tilde{\mathscr{H}}=\mathscr{H}\cap\{\bm{X}~|~\|\bm{X}\|_{\infty}\leq B\}$ with $B$ a very large number, and \eqref{eq:alg_PWGD} is slightly modified by replacing $\mathscr{H}$ by $\tilde{\mathscr{H}}$. Then all the conditions in Theorem \ref{thm:convPAMA} can still be verified. Since $\|(\bm{L}_k,\bm{H}_k)\|_F\not\to\infty$, we must have $(\bm{L}_k,\bm{H}_k)$ converges.

\subsection{Acceleration by a FISTA Scheme}
In this subsection, we propose a scheme to accelerate the convergence of the PWGD algorithm \eqref{eq:alg_PWGD}. Our scheme borrows from the fast iterative shrinkage-thresholding algorithm (FISTA) \cite{BT:SIIMS:09}, which has been proven to be efficient in minimizing the sum of two convex functions with one having a Lipschitz continuous gradient. The basic idea is to use a specific linear combination of two successive iterates. Although our problem is non-convex, we still employ the linear combination scheme in FISTA for our model.

Our PWGD with FISTA scheme, called PWGD-FISTA, is constructed as follows: Given $k_0 = 1$, we generate $\{\bm{L}_t,\bm{H}_t\}$ by
\begin{equation}\label{eq:alg_PWGD-FISTA}
 \left\{ \begin{array}{l l}
      \bm{L}_{t+1}\in \mathcal{P}_{\mathscr{R}^R_\mathbb{C}}(\bm{L}_t-\delta_1(\bm{L}_t-\tilde{\bm{H}}_t)),\\
     \bm{H}_{t+1} \in \mathcal{P}_{\mathscr{H}}(\bm{H}_t -\delta_2(\tilde{\bm{H}}_t-\bm{L}_{t+1})),\\
     \ k_{t+1}=\frac{\sqrt{1+4k_t^2}+1}{2},\\
     \ \tilde{\bm{H}}_{t+1}=\bm{H}_{t+1}+\frac{k_t-1}{k_{t+1}}(\bm{H}_{t+1}-\bm{H}_t)
   \end{array} \right.
\end{equation}
Since $\mathscr{H}$ is an affine subspace, the linear combination in the last line of \eqref{eq:alg_PWGD-FISTA} does not change the feasibility of $\tilde{\bm{H}}_{t+1}$, i.e., $\tilde{\bm{H}}_{t+1}\in\mathscr{H}$. This guarantees that the computational complexity and required storage of Step 1 and Step 2 in the PWGD-FISTA algorithm are the same as that in the PWGD algorithm \eqref{eq:alg_PWGD}. Also, the computational effort in Step 3 and Step 4 of \eqref{eq:alg_PWGD-FISTA} is negligible compared with that in Step 1 and Step 2. Therefore, the PWGD-FISTA algorithm preserves the computational simplicity of the PWGD algorithm. As we will see in the numerical experiments section, the PWGD-FISTA Algorithm converges faster than the PWGD algorithm.

\section{Numerical experiments}
In this section, we use numerical experiments to demonstrate the effectiveness and efficiency of our proposed algorithms.

\subsection{Phase Transition}
We first illustrate that our proposed algorithm is able to recovery spectrally sparse signals from their very limited time domain samples. We fix the dimension of the signal to be $127$ (i.e. $N=64$), and we vary the sparsity $R$ and the number of samples $M$. For each $(R,M)$ pair, $100$ Monte Carlo trials were conducted. For each trial, the true signal is synthesized by randomly generating the true frequencies $\true{f}_k$'s and magnitudes $\true{d}_k$'s, which are independently uniformly distributed on $[0,1)$ and the unit circle respectively.
We then get $M$ samples uniformly at random. The PWGD is executed by setting the parameters $\delta_1 = \delta_2 = 0.9999$, and we stop the algorithm when $\|\bm{H}_{t+1}-\bm{H}_{t}\|_F/\|\bm{H}_{t}\|_F\leq 10^{-4}$. Signal recovery for each trial is considered successful if the relative error satisfies $\|\hat{\bm{x}}-\true{\bm{x}}\|_2/\|\true{\bm{x}}\|_2\leq 5\times 10^{-3}$, where $\hat{\bm{x}}$ denotes the solution returned by the PWGD. Figure \ref{fig:phase} illustrates the results of the Morte Carlo experiments. Here the horizontal axis corresponds to the number $M$ of the samples (i.e. the size of the observation location set $\Theta$), while the vertical axis corresponds to the sparsity level $R$. The empirical success rate is reflected by the color of each cell. It can be seen from the figure that our proposed algorithm has a high rate of successful recovery if $M$ exceeds than certain thresholds for a given $R$.

\begin{figure}[H]
\centering
\includegraphics[width=10cm, height=8cm]{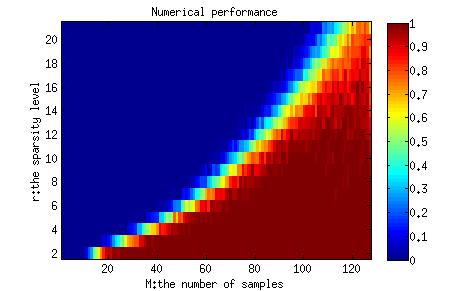}
\caption{Phase transition for successful recovery rate, where frequency locations are randomly generated. The horizontal axis stands for the number of samples, and the vertical axis represents the sparsity $R$. The demonstrated empirical success rate is calculated by averaging over 100 Monte Carlo trials.}\label{fig:phase}
\end{figure}

\subsection{Signals of large dimension}
Next we demonstrate that our proposed algorithm is able to recover signals of large scale, and compare it with the Enhance Matrix Completion (EMaC) in \cite{CC:TIT:14}. As we have argued, different from existing convex optimization based methods such as EMaC, our proposed algorithm is able to work with high-dimensional spectrally sparse signals. In Table \ref{tab:largescale}, the elapsed time for signals of different dimensions are listed. For our algorithm, we use the same settings as in the previous section. For EMaC algorithm, we used CVX software to solve it. From the table, we can see that PWGD can greatly speed up the the signal recovery for moderate dimensions and also work well for signals of high dimensions.

\begin{table}[h]
\begin{center}
\begin{tabular}{c|c|c}
\hline\hline
          & PWGD       &  EMaC      \\ \hline
the signal with $N=51,R=1,M=10$      & 0.34  & 46.8 \\
the signal with $N=51,R=3,M=20$      & 0.46  & 58.0 \\
the signal with $N=101,R=5,M=40$     & 0.95  & out of memory \\
the signal with $N=501,R=5,M=100$    & 12.7 & out of memory  \\
the signal with $N=2501,R=13,M=500$    & 133  & out of memory\\
the signal with $N=2501,R=25,M=1000$    & 91.4&  out of memory \\
the signal with $N=5001,R=20,M=1000$ &  645& out of memory \\
the signal with $N=5001,R=31,M=2000$ &   403 &  out of memory \\ \hline
\end{tabular}
\caption{Elapsed time in seconds for signals of different dimensions.}\label{tab:largescale}
\end{center}
\end{table}

\subsection{Acceleration by a FISTA-like Scheme}
Figure \ref{fig:CompConv} depicts the convergence curves of PWGD and PWGD-FISTA. We see clearly that the PWGD-FISTA Algorithm converges faster than the PWGD algorithm. Roughly, the PWGD-FISTA needs only $2/3$ number of iterations that PWGD requires to get solutions of the same accuracy.

\begin{figure}[H]
\begin{subfigure}{.45\textwidth}
  \centering
  \includegraphics[width=.8\linewidth]{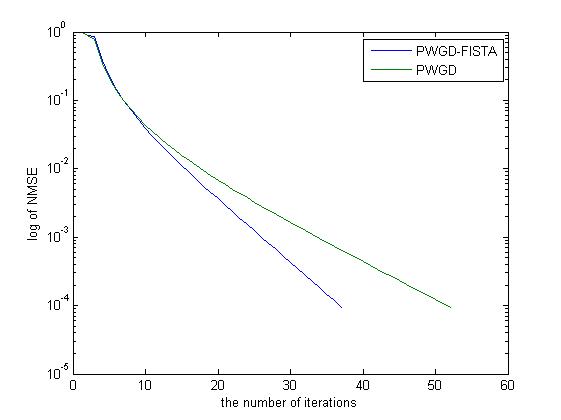}
  \caption{$N=501,~R=8,~M=200$}
\end{subfigure}
\begin{subfigure}{.45\textwidth}
  \centering
  \includegraphics[width=.8\linewidth]{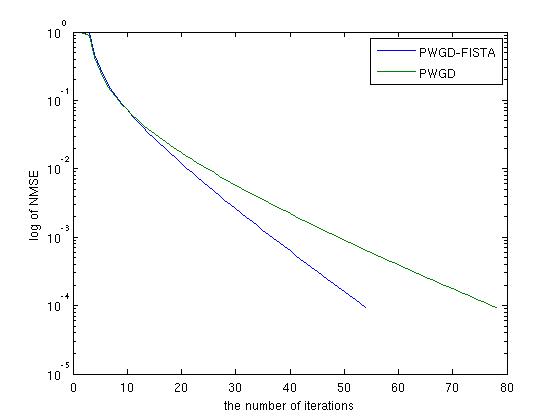}
  \caption{$N=5001,~R=20,~M=1000$}
\end{subfigure}\\
\caption{Convergence rate comparison of PWGD and PWGD-FISTA .}
\label{fig:CompConv}
\end{figure}

\section{Conclusion}
In this paper, a fast iterative algorithm is proposed for recovering spectrally sparse signals whose frequencies can be any values in the continuous domain $[0,1)$ from a small amount of time domain samples. Different from existing algorithms, our proposed algorithm is able to deal with signals of large dimension. Inspired by the scheme in FISTA, we also provided an acceleration of the proposed algorithm. In the future, we will extend our algorithms to signal recovery from noisy samples and signals with multivariate frequencies.

\def\cprime{$'$}

\end{document}